\documentclass[lnicst]{svmultln}
\usepackage{times}
\usepackage{verbatim}
\usepackage{amsmath}
\usepackage{amssymb}

\begin{document}

\mainmatter

\title{False-name-proofness with Bid Withdrawal}

\author{Mingyu Guo\inst{1} \and Vincent Conitzer\inst{2}}
\institute{
University of Liverpool, Department of Computer Science, Liverpool, Merseyside, UK \newline \email{Mingyu.Guo@liverpool.ac.uk} \and
Duke University, Department of Computer Science, Durham, NC, USA \newline \email{conitzer@cs.duke.edu}}

\date{}
\maketitle

\begin{abstract} We study a more powerful variant of false-name manipulation in
Internet auctions: an agent can submit multiple false-name bids, but then,
once the allocation and payments have been decided, withdraw some of her
false-name identities (have some of her false-name identities refuse to
pay).  While these withdrawn identities will not obtain the items they won,
their initial presence may have been beneficial to the agent's other
identities.  We define a mechanism to be {\em false-name-proof with
  withdrawal (FNPW)} if the aforementioned manipulation is never
beneficial.  FNPW is a stronger condition than false-name-proofness (FNP).

 \end{abstract}

\section{Introduction}

With the rapid development of electronic commerce, Internet auctions have
become increasingly popular over the
years~\cite{Monderer98:Optimal,Wurman98:Michigan,Sandholm96:Vickrey}.  Unlike
traditional auctions, typical Internet auctions pose no geographical
constraint. That is, sellers and bidders from all over the world can
participate in an Internet auction remotely over the Internet, without having
to physically attend the auction event.  For sellers, this reduces the cost of
running an auction.  For bidders, this lowers the entry cost.  Effectively, in
an individually rational auction mechanism (a mechanism that guarantees
nonnegative utilities for the agents), a bidder, at worst, loses
nothing (but time) by participating in an auction.  On the one hand, this
encourages more bidders to join the auction, which potentially leads to higher
revenue for the seller, as well as a higher social welfare for the bidders.  On
the other hand, it enables the bidders to manipulate by submitting multiple
bids via multiple fictitious identities ({\em e.g.}, user accounts linked to
different e-mail addresses).

The line of research on preventing manipulation via multiple fictitious
identities in Internet auctions was explicitly framed by the groundbreaking
work of Yokoo {\em et al.}~\cite{Yokoo04:effect}.  Extending {\em
	strategy-proofness}---the concept of ensuring that it is always in a
	bidder's best interest to report her valuation function truthfully---the
	authors define an auction mechanism to be {\em false-name-proof (FNP)} if the
	mechanism is not only strategy-proof, but also, under this mechanism, an
	agent cannot benefit from submitting multiple bids under false names
	(fictitious identities).  The authors also extended the revelation
	principle~\cite{Myerson81:Optimal} to incorporate false-name-proofness.
	That is (roughly stated), in settings where false-name bids are possible,
	it is without loss of generality to focus only on false-name-proof
	mechanisms.  

Focusing primarily on combinatorial auctions, this paper continues the line of
research on false-name-proofness by considering an even more powerful variant
of false-name manipulation: an agent can submit multiple false-name bids, but
then, once the allocation and payments have been decided, withdraw some of her
false-name identities (have some of her false-name identities refuse to pay).
While these withdrawn identities will not obtain the items they won, their
initial presence may have been beneficial to the agent's other identities, as
shown in the following example:  

\begin{example} There are three single-minded
 agents $1,2,3$ and two items
$A,B$. Agent $1$ bids $4$ on $\{A,B\}$.  Agent $2$ bids $2$ on $\{B\}$.  Let us
analyze the strategic options for agent $3$, who is single-minded on $\{A\}$, with
valuation $1$.  (That is, $\forall S\subseteq \{A,B\}$, agent $3$'s valuation
for $S$ is $1$ if and only if $\{A\}\subseteq S$.) The mechanism under
consideration is the VCG mechanism.  

If agent $3$ reports truthfully, then she wins nothing and pays nothing. Her
resulting utility equals $0$.  

If agent $3$ attempts ``traditional'' false-name manipulation, that is,
submitting multiple false-name bids, and honoring 
all of them at the end, then
her utility is still at most $0$: if $3$ wins both items with one identity,
then she has to pay at least $4$ (while her valuation for the items is only
$1$); if $3$ wins both items with two identities (one item for each identity),
then the identity winning $\{B\}$ has to pay at least $2$; if $3$ wins only
$\{B\}$ or nothing, then her utility is at most $0$; 
if $3$ wins only $\{A\}$
(in which case $\{B\}$ has to be won by agent $2$), 
then $3$'s winning identity's payment equals
the other identities' overall valuation for $\{A,B\}$ (at least $4$), minus
$2$'s valuation for $\{B\}$ (which equals $2$).  That is, in this case, $3$ has to
pay at least $2$. So, overall, $3$'s utility is at most $0$ if she honors all her
bids.

However, agent $3$ can actually benefit from submitting multiple false-name
bids, as long as she can withdraw some of them.  For example, $3$ can use two
identities, $3a$ and $3b$.  $3a$ bids $1$ on $\{A\}$. $3b$ bids $4$ on $\{B\}$.
At the end, $3a$ wins $\{A\}$ for free, and $3b$ wins $\{B\}$ for $2$. If $3$ can
withdraw identity $3b$ ({\em e.g.}, by never checking that e-mail account
anymore), never making the payment and never collecting $\{B\}$, then, she has
obtained $\{A\}$ for free, resulting in a utility of $1$.
\end{example}

If we wish to guard against manipulations like the above, we need to extend the
false-name-proofness condition.  We refer to the new condition as {\em
false-name-proofness with withdrawal (FNPW)}. It requires that, regardless
of what other agents do, an agent's optimal strategy is to report
truthfully using a single identity, even if she has the option to submit
multiple false-name bids, and withdraw some of them at the end of the
auction.

Whether this stronger version of false-name-proofness is more or less
reasonable than the original version depends on the context.  First of all, bid
withdrawal is a common example of strategic bidding and has been observed in
real-life auctions such as the FCC Spectrum
Auctions~\cite{Porter99:The,Rothkopf91:On,Holland05:Robust,Cramton97:The}.
That is, bid withdrawal is a threat that we can not ignore.  Then, in some
sense, false-name manipulation and bid withdrawal go hand in hand---in
highly-anonymous settings where agents can easily create multiple fictitious
identities, agents generally can also easily discard their fictitious
identities (without needing to worry about reputations or lawsuits).  From
these perspectives, it is reasonable to study FNPW mechanisms.  In any case,
FNPW is a useful conceptual tool for analyzing false-name-proof
mechanisms.  Indeed, this paper also contributes to the research on
false-name-proofness in the traditional sense. Since FNPW is stronger
than FNP, the results in this paper, such as the mechanism proposed,
should be of interest in the FNP context as well.
  
On the other hand, people may argue that there are two natural measures to prevent bid withdrawal, as
shown below:

\begin{itemize}
\item Require each identifier to pay a certain amount of deposit before the auction begins.
If an identifier withdraws, then her deposit is forfeited. 
\item If any identifier withdraws, then we reallocate according to certain reallocation rules ({\em e.g.}, run the original auction again).
\end{itemize}

However, while discouraging bid withdrawal, both measures lead to other problems.
The problem of the first measure is that small deposits may not be enough
to discourage bid withdrawal, while large deposits may significantly discourage
participation ({\em e.g.}, an agent may be willing to sell some of her assets
			to gather enough cash to pay for the items once she wins, but may
			not be willing to sell her assets just for paying for the deposit).
Also, when there are many participants and few potential winners, it is
unnecessary and costly to collect everyone's deposit.  The problem
of the second measure is that 1) After reallocation, an agent who did not
withdraw may end up with worse result than before. That is, an agent may
be punished by others' faults.
2) It may take some time before the auctioneer figures out that a reallocation
is needed ({\em e.g.}, it is the transaction deadline, but there are still
			winners who haven't paid). That is, reallocation may be late
({\em e.g.}, the items expired or the bids expired). 3) An agent may submit
many false-name bids. After the initial result comes out, she may get some idea of the
other agents' bids. Then, this agent can withdraw all her false-name bids,
	  except for the one bid that is the best response to the other agents' bids.\footnote{This type of manipulation was studied in \cite{Rothkopf91:On}.}  In light
	  of the above, in this paper, we focus on mechanisms that discourage bid
	  withdrawal in the first place, without resorting to charging deposit or reallocation.


\section{Related Research and Contributions}

The main topic of this paper is the comparison between FNP and FNPW.  FNPW is
certainly more restricted, but this doesn't necessarily mean that FNPW is less
interesting. It could be that FNPW is only slightly more restricted and much
more structured. This is what we are trying to find out. Our results on FNPW
and their comparison against previous results on FNP are summarized below:

Yokoo~\cite{Yokoo03:Characterization} and Todo {\em et
al.}~\cite{Todo09:Characterizing} characterized the payment rules and the
allocation rules of FNP mechanisms in general combinatorial auctions,
respectively.  We present
similar results on the characterization of FNPW mechanisms.  As was
in the case of FNP, the characterization of FNPW {\em payment} rules
is useful for proving a given mechanism to {\em be} FNPW, while the
characterization of FNPW {\em allocation} rules is useful for
proving a given mechanism to {\em be not} FNPW.

With our characterizations, we are able to prove whether an existing FNP
mechanism is FNPW or not.\footnote{It should be noted that the
	characterizations are helpful, but definitely not {\em necessary} for
	  proving whether an existing FNP mechanism is FNPW or not. For example, we
	  could always try to use counter examples to show that a mechanism is not
	  FNPW.} There are three known FNP mechanisms for general combinatorial
	  auction settings. These are the Set
	  mechanism~\cite{Yokoo03:Characterization}, the Minimal Bundle (MB)
		mechanism~\cite{Yokoo03:Characterization}, and the Leveled Division Set
		(LDS) mechanism~\cite{Yokoo01:Robust}.\footnote{A very recent
			paper~\cite{Iwasaki10:Worst} introduced another mechanism called
			  the ARP mechanism.  However, this mechanism requires the
			  additional restriction that agents are single-minded.}  We show
			  that both Set and MB are FNPW, while LDS is not.  We also show
			  that the VCG mechanism is FNPW if and only if the type space
			  satisfies the submodularity condition (with a minor assumption).  Previously, Yokoo {\em et
				  al.}~\cite{Yokoo04:effect} showed that the submodularity
				  condition is sufficient for the VCG mechanism to be FNP.

We then compare the worst-case efficiency ratios of FNP and FNPW mechanisms.
Iwasaki {\em et al.}~\cite{Iwasaki10:Worst} showed that, under a
minor condition, the worst-case efficiency
ratio of any feasible FNP mechanism is at most $\frac{2}{m+1}$.
We show that under the same condition, the 
worst-case efficiency
ratio of any feasible FNPW mechanism is at most $\frac{1}{m}$.
\begin{footnotesize}
\begin{table}
\caption{FNP v.s. FNPW}
\begin{tabular}{|c|c|c|c|c|c|c|c|}
	\hline
 & Characterization of & Characterization of & Set, MB, LDS& Worst-case \\
&payment rules & allocation rules & VCG w. submodularity & efficiency ratio\\
\hline
FNP & NSA\cite{Yokoo03:Characterization} &  Sub-additivity\cite{Todo09:Characterizing} &Yes\cite{Yokoo03:Characterization},
Yes\cite{Yokoo03:Characterization},Yes\cite{Yokoo01:Robust}&$\frac{2}{m+1}$\cite{Iwasaki10:Worst}\\
	&									 & Weak-monotonicity & Yes\cite{Yokoo04:effect} & \\ 
\hline
FNPW & NSAW & Sub-additivity & Yes, Yes, No & $\frac{1}{m}$\\
	& (S-NSAW) & Weak-monotonicity & Yes & \\
	&		 & Withdrawal-monotonicity & &\\
\hline
\end{tabular}
\end{table}
\end{footnotesize}

At the end, we propose the {\em maximum marginal value item pricing (MMVIP)}
mechanism.  We show that MMVIP is FNPW and exhibits some desirable
properties.  Since FNPW is stronger than FNP, MMVIP also adds to the set of known FNP mechanisms.

Finally, in Appendix~\ref{sec:amd}, we propose an (exponential-time) automated mechanism
design technique that transforms any feasible mechanism into a FNPW mechanism,
and prove some basic properties about this technique.  We also give a
characterization of FNP(W) social choice rules in Appendix~\ref{sec:social}.

\section{Formalization}
\label{sec:formalization}

We will use the following notation:

\begin{itemize}
\item 
$N=\{1,2,\ldots,n\}$: the set of agents 

\item
$G=\{1,2,\ldots,m\}$: the set of items 

\item
$\Theta$: the type space of each agent

\item 
$\theta_i \in \Theta$: agent $i$'s reported type (since we consider only strategy-proof mechanisms, when there is no ambiguity, we also use $\theta_i$ to denote $i$'s true type)

\item 
$-i$: the set of agents other than agent $i$ 

\item
$\theta_{-i}\in \Theta^{n-1}$: types reported by agents other than agent $i$
\end{itemize}

We study combinatorial auction settings satisfying the following assumptions:

\begin{itemize} 
	
\item Each agent has a {\em quasi-linear} utility function. That is, there
exists a function $v$ (determined by the setting) such that if an agent with true
type $\theta\in \Theta$ ends up with bundle $B\subset G$ and payment $p\in
\mathbb{R}$,
then her utility equals $v(\theta, B)-p$.

\item $\forall \theta\in \Theta$, we have $v(\theta, \emptyset)=0$. 

\item $\forall B_1\subseteq B_2\subseteq G$, $\forall \theta\in \Theta$, we
	have $v(\theta, B_1)\le v(\theta,B_2)$. That is, there is {\em free
	disposal}.

\item 
An agent can have any valuation function satisfying the above conditions.
That is, we are dealing with {\em
rich domains}~\cite{Bikhchandani06:Weak}.  It should be noted that in
Section~\ref{sec:type}, we study how restrictive the type space has to be in order
for the VCG mechanism to be FNPW. 
That is, we do not have the rich-domain
assumption in Section~\ref{sec:type}, which is an exception. 

\end{itemize}

A mechanism consists of an allocation rule $X: (\Theta, \Theta^{n-1})\to
\mathcal{P}(G)$ and a payment rule $P: (\Theta, \Theta^{n-1})\to \mathbb{R}$.
$X(\theta_i,\theta_{-i})$ is the bundle agent $i$ receives when reporting
$\theta_i$ (when the other agents report $\theta_{-i}$).
$P(\theta_i,\theta_{-i})$ is the payment agent $i$ has to make when
reporting $\theta_i$ (when the other agents report $\theta_{-i}$).  When there
is no ambiguity about the other agents' types, we simply use $X(\theta_i)$ and
$P(\theta_i)$ in place of $X(\theta_i,\theta_{-i})$ and
$P(\theta_i,\theta_{-i})$.  

Throughout the paper, we only consider mechanisms satisfying the following
conditions:

\begin{itemize} 

\item {\em Strategy-proofness:} $\forall \theta_i,\theta_i',\theta_{-i}$, we have
$v(\theta_i,X(\theta_i))-P(\theta_i)
\ge
v(\theta_i,X(\theta_i'))-P(\theta_i')$. That is, if an
agent uses only one identity, then truthful reporting is a dominant strategy.

\item {\em Pay-only:} $\forall \theta_i,\theta_{-i}$, we have $P(\theta_i)\ge 0$.

\item {\em Individual rationality:} $\forall \theta_i,\theta_{-i}$, we have
$v(\theta_i, X(\theta_i))-P(\theta_i)
\ge 0$. That is, if an agent reports
truthfully, then her utility is
guaranteed to be nonnegative. This condition also implies that
if an agent does not win any items, or has valuation $0$ for all the items, then
her payment must be $0$.

\item {\em Consumer sovereignty:} $\forall \theta_{-i}$, $\forall B\subseteq
G$, there exists $\theta_i\in \Theta$ such that
$X(\theta_i,\theta_{-i})\supseteq B$. That is, 
no matter what the other agents bid,
an agent can always win any bundle
(possibly at the cost of a large payment).

\item {\em Determinism and symmetry:} We only consider deterministic mechanisms that
are symmetric over 
both the agents and the items (except for
ties).  
\end{itemize}

Yokoo~\cite{Yokoo03:Characterization} showed that in our setting,
the mechanisms satisfying the above conditions coincide with the {\em (anonymous)
price-oriented, rationing-free (PORF)} mechanisms. Similar price-based
representations have also been presented by others, including \cite{Lavi03:Towards}.
The PORF mechanisms work as
follows:

\begin{itemize}
	\item The agents submit their reported types.

	\item The mechanism is characterized by a price function $\chi:
		\mathcal{P}(G)
		\times \Theta^{n-1} \to [0,\infty)$.  For any agent $i$,
		for any multiset $\theta_{-i}$ of types reported 
		by the other agents, for any set of items $S\subseteq G$, $\chi(S,\theta_{-i})$ is
		the price of $S$ offered to $i$ by the mechanism. That is, $i$ can purchase
		$S$ at a price of $\chi(S,\theta_{-i})$.  $\forall \theta_{-i}$, we have
		$\chi(\emptyset,\theta_{-i})=0$.  That is, the price of nothing is always
		zero. $\forall \theta_{-i}$, $\forall S_1\subseteq S_2\subseteq G$, we have
		$\chi(S_1,\theta_{-i})\le \chi(S_2,\theta_{-i})$.  That is, a larger
		bundle always
		has a higher (or the same) price.

	\item 
		The mechanism will select a bundle for agent $i$ that is optimal for her given
		the prices, that is, the bundle chosen for $i$ is in 
		$\arg\max_{S\subseteq G} \{v(\theta_i,S)-\chi(S,\theta_{-i})\}$.
The agent then pays the price for this bundle.

	\item Naturally, the mechanism must ensure that no item is allocated to two different agents.
This involves setting prices carefully, as well as breaking ties.
\end{itemize}

Since all {\em feasible} mechanisms (mechanisms that satisfy the desirable
conditions in our setting) are PORF mechanisms, besides using $X$ (the
allocation rule) and $P$ (the payment rule) to refer to a mechanism, we can
also use the price function $\chi$ to refer to a mechanism, namely, the PORF
mechanism with price function $\chi$.\footnote{Technically, there can be 
multiple PORF mechanisms
with the same price function due to tie-breaking, but this will generally not
be an issue.}

In the remainder of this section, we formally define the traditional
false-name-proofness (FNP) condition,
as well as our new false-name-proofness with withdrawal
(FNPW) condition.

\begin{definition} {\bf FNP.}
A mechanism characterized by allocation rule $X$ and payment rule $P$ is FNP if and only if it satisfies the following:

\begin{center}
$\forall \theta_i$, $\forall \theta_{i1},\theta_{i2},\ldots,\theta_{ik}$, $\forall \theta_{-i}$, we have 

$v(\theta_i,X(\theta_i,\theta_{-i}))-P(\theta_i,\theta_{-i})\ge
v(\theta_i,\bigcup\limits_{j=1}^k X(\theta_{ij},\theta_{-i}\cup(\bigcup\limits_{t\neq j}\theta_{it})))-\sum\limits_{j=1}^k P(\theta_{ij},\theta_{-i}\cup(\bigcup\limits_{t\neq j}\theta_{it}))$
\end{center}

That is, truthful reporting using a single
identifier is always better than submitting multiple false-name bids.
\end{definition}

\begin{definition}
\label{def:fnpw}
{\bf FNPW.}
A mechanism characterized by allocation rule $X$ and payment rule $P$ is FNPW if and only if it satisfies the following:

\begin{center}
$\forall \theta_i$, $\forall \theta_{i1},\theta_{i2},\ldots,\theta_{ik}$, 
$\forall \theta_{i1}',\theta_{i2}',\ldots,\theta_{iq}'$,
$\forall \theta_{-i}$, we have 

$v(\theta_i,X(\theta_i,\theta_{-i}))-P(\theta_i,\theta_{-i})\ge$
$v(\theta_i,\bigcup\limits_{j=1}^k X(\theta_{ij},\theta_{-i}\cup(\bigcup\limits_{t\neq j}\theta_{it})\cup(\bigcup\theta_{it}')))-\sum\limits_{j=1}^k P(\theta_{ij},\theta_{-i}\cup(\bigcup\limits_{t\neq j}\theta_{it})\cup(\bigcup\theta_{it}'))$
\end{center}
That is, truthful reporting using a single identifier is always better than submitting multiple
false-name bids and then withdrawing some of them.
\end{definition}

It is easy to see that FNPW is equivalent to FNP plus the following condition: 
an agent's utility for reporting truthfully does not increase if we add another agent.

\section{Characterization of FNPW mechanisms}
\label{sec:characterization}

Yokoo~\cite{Yokoo03:Characterization} and Todo {\em et
al.}~\cite{Todo09:Characterizing} characterized the payment rules (the price
functions in the PORF representation) and the allocation rules of FNP
mechanisms, respectively.  In this section, we present similar results on the
characterization of FNPW mechanisms.

\subsection{Characterizing FNPW payments}
\label{sec:nsaw}

We recall that in our setting, a feasible mechanism corresponds to a PORF
mechanism, characterized by a price function $\chi$.
Yokoo~\cite{Yokoo03:Characterization} gave the following sufficient and
necessary condition on $\chi$ for the mechanism characterized by $\chi$ to be
FNP. 

\begin{definition} {\bf No superadditive price increase (NSA)~\cite{Yokoo03:Characterization}.} Let $O$ be an
	arbitrary set of agents.\footnote{In a slight abuse of language, we also use ``a set of agents''
	to refer to the types reported by this set of agents.}
	We run mechanism $\chi$ (a PORF mechanism
	characterized by price function $\chi$) for the agents in $O$.  Let $Y$ be an
	arbitrary subset of $O$.  Let $B_i$ ($i\in Y$) be the set of items agent
	$i$ obtains.  We must have 
	$\sum\limits_{i\in Y}\chi(B_i,O-\{i\})\ge
	\chi(\bigcup\limits_{i\in Y}B_i,O-Y)$.  
\end{definition}

By modifying the NSA condition, we get the following sufficient and necessary
condition on $\chi$ for mechanism $\chi$ to be FNPW.

\begin{definition} {\bf No superadditive price increase with withdrawal (NSAW).} Let $O$ be an
	arbitrary set of agents. 
	We run mechanism $\chi$ for the agents in $O$.  Let $Y$ and $Z$ be two
	arbitrary nonintersecting
	subsets of $O$.  Let $B_i$ ($i\in Y$) be the set of items agent
	$i$ obtains.  We must have 
	$\sum\limits_{i\in Y}\chi(B_i,O-\{i\})\ge
	\chi(\bigcup\limits_{i\in Y}B_i,O-Y-Z)$.  
\end{definition}




\begin{theorem} Mechanism $\chi$ is FNPW
if and only if $\chi$ satisfies the NSAW condition.  
\label{th:NSAW}
\end{theorem}

\subsection{A sufficient condition for FNPW}

The NSAW condition in Section~\ref{sec:nsaw} leads to the following sufficient
condition for mechanism $\chi$ to be FNPW.

\begin{definition} {\bf Sufficient condition for no superadditive price increase with withdrawal (S-NSAW).} 
Let $O$ be an
	arbitrary set of agents. S-NSAW holds if we have both of the following conditions:
	\begin{itemize}
		\item {\bf Discounts for larger bundles (DLB).}
	$\forall S_1,S_2\subseteq G$ with $S_1\cap S_2=\emptyset$, 
$\chi(S_1,O)+\chi(S_2,O)\ge \chi(S_1\cup S_2,O)$.
	That is, the sum of the prices of two disjoint sets of items must be at least the price of the joint set.
\item {\bf Prices increase with agents (PIA).}
	$\forall S\subseteq G$, for any agent $a$ that is not in $O$, 
	$\chi(S,O\cup\{a\})\ge \chi(S,O)$.
	That is, from the perspective of agent $i$, if another agent joins in, then
	the price $i$ faces for any set of items  must (weakly) increase.
	\end{itemize}
\end{definition}
	
\begin{proposition} Mechanism $\chi$ is FNPW if
	$\chi$ satisfies S-NSAW.
	\label{prop:s-nsaw}
\end{proposition}

S-NSAW is a cleaner, but more restrictive condition than NSAW.
(To see why, note that even if DLB does not hold,
NSA may still hold: even if $\chi(S_1,O)+\chi(S_2,O) < \chi(S_1\cup S_2,O)$, it may be the case that 
by putting separate bids on $S_1$ and $S_2$, each of these bids makes the price for the other bundle go up, so that the result is still more expensive than buying $S_1\cup S_2$ as a single bundle.)
We find it easier to
use S-NSAW to prove that a mechanism is FNPW (rather than using the more complex
NSAW condition).\footnote{However, S-NSAW cannot be used to prove that a
mechanism is {\em not} FNPW, because it is a more restrictive condition.} Let us recall the three existing FNP
mechanisms (for general combinatorial auction settings): the Set mechanism, the MB
Mechanism, and the LDS mechanism. 
With the help of S-NSAW, we can prove that both Set and MB are FNPW.

\begin{proposition}
Both the Set mechanism and the MB mechanism satisfy the S-NSAW condition.
Hence, they are FNPW.
\end{proposition}

The Set mechanism simply combines all the items into a grand bundle. The grand
bundle is then sold in a Vickrey auction.  The MB (Minimal Bundle) mechanism
builds on the concept of minimal bundles. A set of items $S$
($\emptyset\subsetneq S\subset G$) is called a {\em minimal bundle} for agent
$i$ if and only if $\forall S'\subsetneq S$, $v(i,S)>v(i,S')$.  Under the MB
mechanism, the price of a bundle $S$ an agent faces is equal to the highest
valuation value of a bundle, which is minimal and conflicting with $S$.
Generally, MB coincides with Set, because usually the grand bundle is a minimal
bundle for every agent (any smaller bundle usually gives at least slightly
lower utility).  The proof of the above proposition is straightforward.

We will also use S-NSAW to prove that the MMVIP mechanism that we propose
(Section~\ref{sec:mmvip}) is FNPW.  The automated mechanism design technique
for generating FNPW mechanisms (Appendix~\ref{sec:amd}) is also based on
S-NSAW.

\subsection{Characterizing FNPW allocations}

Todo {\em et al.}~\cite{Todo09:Characterizing} gave the following characterization
of the allocation rules of FNP mechanisms.
We recall that $X(\theta_i,\theta_{-i})$ is the set of items that agent $i$ wins
if her reported type is $\theta_i$ and the reported types of the other agents
are $\theta_{-i}$.  To simplify notation, we use $X(\theta_i)$ in
place of $X(\theta_i,\theta_{-i})$ when there is no risk of ambiguity.

\begin{definition}
	{\bf Weak-monotonicity~\cite{Bikhchandani06:Weak}.}
	$X$ is weakly monotone if $\forall \theta_i,\theta_i',\theta_{-i}$,
	we have
	\begin{center}
	$v(\theta_i,X(\theta_i))-
	v(\theta_i,X(\theta_i'))\ge
	v(\theta_i',X(\theta_i))-
	v(\theta_i',X(\theta_i'))$.
	\end{center}
\end{definition}

\begin{definition}
	{\bf Sub-additivity~\cite{Todo09:Characterizing}.}
	$\forall \theta_i$, $\forall \theta_i'$, 
	$\forall \theta_{i1},\theta_{i2},\ldots,\theta_{ik}$,
	$\forall \theta_{i1}',\theta_{i2}',\ldots,\theta_{ik}'$,
	$\forall \theta_{-i}$,
we have	the following:

	\begin{center}

	$X(\theta_i)=
	\bigcup\limits_{l=1}^k X_{+I^k_{-l}}(\theta_{il})$,
	$v(\theta_i',X(\theta_i'))=0$

	$X_{+I^k_{-l}}(\theta_{il}')\supseteq 
	X_{+I^k_{-l}}(\theta_{il})$,
	$v(\theta_{il}',X_{+I^k_{-l}}(\theta_{il}'))=
	v(\theta_{il}',X_{+I^k_{-l}}(\theta_{il}))$

	$\Downarrow$

	$v(\theta_i',X(\theta_i))\le \sum\limits_{l=1}^k v(\theta_{il}',X_{+I^k_{-l}}(\theta_{il}))$.

	Here, $X_{+I^k_{-l}}(\theta_{il})$ is short for 
	$X(\theta_{il},\theta_{-i}\cup(\bigcup\limits_{1\le t\le k,t\ne l}\theta_{it}))$.

	$X_{+I^k_{-l}}(\theta_{il}')$ is short for 
	$X(\theta_{il}',\theta_{-i}\cup(\bigcup\limits_{1\le t\le k,t\ne l}\theta_{it}))$.
	\end{center}
\end{definition}

$X$ is said to be {\em FNP-implementable} if there exists a payment rule $P$ so
that $X$ combined with $P$ constitutes a feasible FNP mechanism.  Todo {\em et
al.}~\cite{Todo09:Characterizing} showed that $X$ is FNP-implementable if and
only $X$ satisfies both 
weak-monotonicity and sub-additivity.

We define allocation rule $X$ to be {\em FNPW-implementable} if there exists a
payment rule $P$ so that $X$ combined with $P$ constitutes a feasible FNPW
mechanism.  
We introduce a third condition called {\em
withdrawal-monotonicity}.
 We prove that $X$ is FNPW-implementable if and only if $X$
satisfies weak-monotonicity, sub-additivity, and 
withdrawal-monotonicity.

\begin{definition} {\bf Withdrawal-monotonicity.}
$\forall \theta_i$, $\forall \theta_{-i}$, $\forall \theta^a$,
$\forall \theta_i^L$,
$\forall \theta_i^U$,
the following holds:

\begin{center}

	$v(\theta_i^L,X(\theta_i^L,\theta_{-i}))=0$,
	$X(\theta_i^U,\theta_{-i}\cup\theta^a)=
	X(\theta_i,\theta_{-i})$

	$\Downarrow$

	$v(\theta_i^L,X(\theta_i,\theta_{-i}))\le v(\theta_i^U,X(\theta_i,\theta_{-i}))$
\end{center}
\end{definition}

\begin{theorem} An allocation rule $X$ is FNPW-implementable if and only if $X$
satisfies weak-monotonicity, sub-additivity, and
withdrawal-monotonicity.  
\label{th:FNPW-implementable}
\end{theorem}

The above theorem implies that a necessary condition for a mechanism to be
FNPW is that its allocation rule $X$ satisfies withdrawal-monotonicity.
That is, one way to prove a (FNP) mechanism to be not FNPW is to generate a lot
of test type profiles, and see whether this mechanism's allocation rule
ever violates withdrawal-monotonicity (this process can be computer-assisted). 
If we find one test type profile that violates withdrawal-monotonicity, then
we are sure that the mechanism under discussion is not FNPW.\footnote{\cite{Todo09:Characterizing} proved two mechanisms to be not FNP, by
	presenting type profiles that violate sub-additivity.}

\begin{proposition} 
	\label{prop:LDS}
	The Leveled Division Set (LDS) mechanism~\cite{Yokoo01:Robust} does not
	satisfy \\
	  withdrawal-monotonicity. That is, LDS is not FNPW in general.
\end{proposition}

\section{Restriction on the type space so that VCG is FNPW}
\label{sec:type}

The VCG mechanism~\cite{Vickrey61,Clarke71:Multipart,Groves73:Incentives}
satisfies several nice properties, including efficiency, strategy-proofness,
individual rationality, and the non-deficit property. Unfortunately, as shown by
Yokoo {\em et al.}~\cite{Yokoo04:effect}, the VCG mechanism is not FNP for
general type spaces. One sufficient condition on the type space for the VCG
mechanism to be FNP is as follows:

\begin{definition} {\bf Submodularity~\cite{Yokoo04:effect}.} For any set of
bidders $Y$, whose types are drawn from $\Theta$, $\forall S_1,S_2\subseteq G$,
we have $U(S_1,Y)+U(S_2,Y)\ge U(S_1\cup S_2,Y)+U(S_1\cap S_2,Y)$.  Here,
$U(S,Y)$ is defined as the total utility of bidders in $Y$, if we allocate
items in $S$ to these bidders efficiently.  \end{definition}

That is, if the type space $\Theta$ satisfies the above condition, then the VCG
mechanism is FNP.  In this section, we aim to characterize type spaces for
which VCG is FNPW. We consider restricted type spaces (that make the VCG
mechanism FNPW) in this section. In other sections, unless specified, we assume that the
rich-domain condition holds.

\begin{theorem} If the type space satisfies the submodularity condition, then
the VCG mechanism is FNPW.  Conversely, if the mechanism is FNPW, and
additionally the type space contains the additive valuations, then the type
space satisfies the submodularity condition.  
\label{th:submodularity}
\end{theorem}

That is, submodularity does not only imply FNP, it actually implies FNPW.
Moreover, unlike for FNP, in the case of FNPW, the converse also holds---if we
allow the additive valuations.

\section{Worst-Case Efficiency Ratio of FNPW Mechanisms}
\label{sec:worst}

Yokoo {\em et al.}~\cite{Yokoo04:effect} proved that in general combinatorial
auction settings, there exists no efficient FNP mechanisms.  
Iwasaki {\em et al.}~\cite{Iwasaki10:Worst} further showed that, under a
minor condition called IIG (described below), the worst-case efficiency
ratio of any feasible FNP mechanism is at most
$\frac{2}{m+1}$.\footnote{Iwasaki {\em et al.}~\cite{Iwasaki10:Worst} also introduced the ARP
  mechanism, whose worst-case efficiency ratio is exactly
  $\frac{2}{m+1}$. However, the ARP mechanism is only FNP for single-minded
agents.
Our next result implies that ARP is not FNPW, even with single-minded
bidders.
}

\begin{definition} {\bf Independence of irrelevant good
(IIG)~\cite{Iwasaki10:Worst}.}  Suppose agent $i$ is winning all the
  items. If we add an additional item that is only wanted by $i$, then $i$
  still wins all the items.
\end{definition}

Given the agents' reported types, the efficiency ratio of a mechanism is
defined as the ratio between the achieved allocative efficiency and the optimal
allocative efficiency (payments are not taken into consideration). The
worst-case efficiency ratio of this mechanism is the minimal such ratio
over all possible type profiles.

\begin{example} {\em The worst-case efficiency ratio of the Set mechanism is at
	least $\frac{1}{m}$~\cite{Iwasaki10:Worst}.} 
Let $v$ be the winning agent's valuation
  for the grand bundle.  
The allocative efficiency of the Set mechanism is $v$.  The optimal
allocative efficiency is at most $mv$, since there are at most $m$ winners
in the optimal allocation, and a winner's valuation (for the items she won)
is at most $v$.
\end{example}

Our next theorem is that $\frac{1}{m}$ is a strict upper bound
on the efficiency ratios of feasible FNPW mechanisms.  That is, the Set
mechanism is worst-case optimal in terms of efficiency ratio.
Of course, this is only a worst-case analysis, which does not
preclude FNPW mechanisms from performing well most of the time.

\begin{theorem}
The worst-case efficiency ratio of any feasible
FNPW mechanism is at most $\frac{1}{m}$, if IIG holds, even with
single-minded bidders.
\label{th:1overm}
\end{theorem}

\section{Maximum Marginal Value Item Pricing Mechanism}
\label{sec:mmvip}

In this section, we introduce a new FNPW mechanism. We recall
that S-NSAW is a sufficient condition for FNPW.  Basically, if a mechanism
satisfies discounts for larger bundles (DLB) and prices increase with agents
(PIA), then we know it is FNPW. Any mechanism that uses {\em item pricing}
satisfies DLB. If the item prices an agent faces also increase with the agents, then
we have a mechanism that also satisifies PIA. MMVIP builds on exactly this item pricing idea.

\begin{definition} {\bf Maximum marginal value item pricing mechanism (MMVIP).}
	Let $O$ be an arbitrary set of agents.
	MMVIP is characterized by the following price function $\chi$.

	\begin{itemize}
		\item
	$\forall S\subseteq G$, $\chi(S,O)=\sum_{s\in
	S}\chi(\{s\},O)$. That is, $\chi$ uses {\em item pricing}.\footnote{It should be noted that the item prices
	faced by different agents are generally different.}

		\item $\forall s\in G$, $\chi(s,O)=\max\limits_{j\in
			O}\max\limits_{S\subseteq G-\{s\}}\{v(j, S\cup\{s\})-v(j,S)\}$.\footnote{In this notation, we 
			assume that	the maximum over an empty set is $0$ (for presentation purpose).  Such notation will also appear later in the paper.
			}
			That
			is, the price
			an agent faces for an item is the maximum possible marginal value that any other agent could have for that item, where the maximum is taken over all possible allocations.
	\end{itemize}
\end{definition}

\begin{proposition}
	MMVIP is feasible and FNPW.
	\label{prop:mmvip}
\end{proposition}

Next, we prove two properties of the MMVIP mechanism.

\begin{proposition} 
\label{prop:mmvipvcg}
Suppose we restrict the domain to additive valuations.  Then, MMVIP
coincides with the VCG mechanism, so that MMVIP=VCG is FNPW and efficient.
\end{proposition}

The above proposition essentially says that, when the agents' valuations are
additive, MMVIP ``does the right thing.''  MMVIP is the only known FNP/FNPW
mechanism with the above property for general combinatorial auctions.
Finally, we have the following proposition about MMVIP.




\begin{proposition} Among all FNPW mechanisms that use item pricing, MMVIP has
minimal payments.  That is, let $\chi$ be the price function of MMVIP. Let
$\chi'$ be a different price function
 corresponding to a different FNPW
mechanism $M$ that also uses item pricing.
We have that there always
exists a set of items $S$ and a set of agents $O$, so that
$\chi'(S,O)>\chi(S,O)$.  
\label{prop:mmvipminimal}
\end{proposition}

\section{Conclusion} 

We studied a more powerful variant of false-name manipulation: an agent can
submit multiple false-name bids, but then, once the allocation and payments
have been decided, withdraw some of her false-name identities.  Since FNPW
is stronger than FNP, this paper also contributes to the research on
false-name-proofness in the traditional sense.

\newpage

\bibliography{mg.bib}
\bibliographystyle{abbrv}

\appendix 

\section{Proofs}

Proof of Theorem~\ref{th:NSAW}:

\begin{proof} We first prove that if $\chi$ satisfies NSAW, then the mechanism
	is FNPW. Let us consider a specific agent $x$.  Let $O-Y-Z$ be the set of
	agents other than herself.  
	Let $Y$ be the set of
	false-name identities $x$ submits and keeps at the end.  Let $Z$ be the set
	of false-name identities $x$ submits but withdraws at the end.  So,
	$O$ is the set of all the identities.  The set of items $x$ receives at the
	end is $\bigcup\limits_{i\in Y}B_i$, where $B_i$ is the bundle won by identity $i$.
	The total price $x$ pays is $\sum\limits_{i\in Y}\chi(B_i,O-\{i\})$.  
	According to NSAW, this price is at least $\chi(\bigcup\limits_{i\in
	Y}B_i,O-Y-Z)$. That is, $x$ 
	would not be any worse off if she just used a single
	identity to buy $\bigcup\limits_{i\in Y}B_i$.
When $x$ uses only one identity, her optimal strategy is to
	report truthfully.  Therefore, if NSAW is satisfied, mechanism $\chi$ is
	FNPW.

Next, we prove that if mechanism $\chi$ is FNPW, then $\chi$ must satisfy NSAW.
Suppose not, that is, suppose there exists some $\chi$ that corresponds to an FNPW mechanism, and there exist
three nonintersecting
sets of agents $Y$, $Z$, and $O-Y-Z$, such that
$\sum\limits_{i\in Y}\chi(B_i,O-\{i\})<\chi(\bigcup\limits_{i\in Y}B_i,O-Y-Z)$,
where $B_i$ is the bundle agent $i$ obtains (when we apply mechanism $\chi$ to the
agents in $O$).  Let us consider a single-minded agent $x$, who values
$\bigcup\limits_{i\in Y}B_i$ at exactly $\chi(\bigcup\limits_{i\in Y}B_i,O-Y-Z)$. If
the set of other agents faced by $x$ is $O-Y-Z$, then $x$ has utility $0$ if
she reports truthfully using a single identifier.  However, if $x$ instead submits multiple false-name
identities $Y+Z$, keeps those in $Y$ and withdraws those in $Z$, then she will
obtain her desired items at a lower price and
end up with positive utility, contradicting the
assumption that $\chi$ is FNPW.
  That is, if NSAW is not satisfied, then $\chi$ is not FNPW.  \end{proof}

Proof of Proposition~\ref{prop:s-nsaw}:

\begin{proof} 
We only need to show that S-NSAW is stronger than NSAW (by Theorem~\ref{th:NSAW}, NSAW is
sufficient (and necessary) for $\chi$ to be FNPW).  Let $\chi$ satisfy S-NSAW.	Let $O$
be an arbitrary set of agents. We run mechanism $\chi$ on the agents in $O$.  We
divide $O$ into three subgroups, $Y$, $Z$, and $O-Y-Z$.  For $i \in Y$, let $B_i$ be the
bundle agent $i$ obtains.  By PIA, we have $\sum_{i\in Y}\chi(B_i, O-\{i\})\ge
\sum_{i\in Y}\chi(B_i,O-Y-Z)$. 
By DLB, we have $\sum_{i\in Y}\chi(B_i,O-Y-Z)\ge \chi(\bigcup\limits_{i\in Y}B_i,O-Y-Z)$.  Combining these inequalities, we can conclude that S-NSAW
implies NSAW.  
\end{proof}

Proof of Theorem~\ref{th:FNPW-implementable}:

\begin{proof} We first prove that if $X$ is FNPW-implementable, then $X$
	satisfies weak-monotonicity, sub-additivity, and
	withdrawal-monotonicity.  If $X$
	is FNPW-implementable, then $X$ is also FNP-implementable.  Hence, $X$
	satisfies both weak-monotonicity and sub-additivity~\cite{Todo09:Characterizing}; 
	only withdrawal-monotonicity remains to be shown.  Let $\chi$ be the (PORF)
	price function corresponding to an FNPW mechanism that allocates according
	to $X$.  We denote $X(\theta_i,\theta_{-i})$ by $S$. Since
	$v(\theta_i^L,X(\theta_i^L,\theta_{-i}))=0$, we have $v(\theta_i^L,S)\le
	\chi(S,\theta_{-i})$ (otherwise, an agent with true type $\theta_i^L$ would be better off
	purchasing $S$).  Since
	$X(\theta_i^U,\theta_{-i}\cup\theta^a)=X(\theta_i,\theta_{-i})=S$, we have
	$v(\theta_i^U, S)\ge \chi(S,\theta_{-i}\cup\theta^a)$ (because an agent with true
	type $\theta_i^U$ is best off buying $S$ when the other agents' types are
	$\theta_{-i}\cup\theta^a$). $\chi$ is FNPW, 
	we must have $\chi(S,\theta_{-i}\cup\theta^a)\ge
	\chi(S,\theta_{-i})$.  Combining all the inequalities, we get $v(\theta_i^U,
	X(\theta_i,\theta_{-i}))\ge v(\theta_i^L,X(\theta_i,\theta_{-i}))$.
	That is, withdrawal-monotonicity is satisfied.

Next, we prove that if $X$ satisfies weak-monotonicity, 
sub-additivity, and
withdrawal-monotonicity, then $X$ is 
FNPW-implementable.  Since $X$ satisfies both
weak-monotonicity and sub-additivity, $X$ is FNP-
implementable~\cite{Todo09:Characterizing}.
 Let $\chi$ be a
(PORF) price function that characterizes an FNP mechanism that allocates
according to $X$. We prove that $\chi$ must also be FNPW.  We only need to
prove that $\chi$ satisfies PIA. That is, under $\chi$, the agents have no
incentives to create indentifiers that will later on be discarded.
Suppose
$\chi$ does not satisfy PIA. Then, there exists a set of agents $O$, an
agent $a$ not in $O$ (where $a$'s type is denoted by $\theta^a$), and some $S\subseteq
G$, such that $\chi(S,O)>\chi(S,O\cup\{a\})$.  Let
$\chi(S,O)-\chi(S,O\cup\{a\})=\beta>0$.  Let $\theta_{-i}$ be the reported types
of the agents in $O$.  Let $i$ be an agent that is single-minded on $S$, with
a very large valuation, so that $X(\theta_i,\theta_{-i})=S$ (we denote agent
$i$'s type by $\theta_i$).  We also construct an agent that is single-minded on $S$,
with valuation $\chi(S,O)-\frac{\beta}{3}$. We denote the type of this agent by
$\theta^L_i$.  We have $X(\theta_i^L,\theta_{-i})=\emptyset$ (she is not willing to pay $\chi(S,O)$ to purchase $S$).  Hence,
$v(\theta_i^L,X(\theta_i^L,\theta_{-i}))=0$.  We construct another agent that
is also single-minded on $S$, with valuation
$\chi(S,O\cup\{a\})+\frac{\beta}{3}$. We denote the type of this agent by
$\theta^U_i$.  We have
$X(\theta_i^U,\theta_{-i}\cup\theta^a)=S=X(\theta_i,\theta_{-i})$.  By
withdrawal-monotonicity, we must have $v(\theta_i^L,X(\theta_i,\theta_{-i}))\le
v(\theta_i^U,X(\theta_i,\theta_{-i}))$.  However,
on the other hand,
$v(\theta_i^L,X(\theta_i,\theta_{-i}))=\chi(S,O)-\frac{\beta}{3} = \chi(S,O\cup\{a\}) + \frac{2\beta}{3}
> \chi(S,O\cup\{a\})+\frac{\beta}{3} = 
v(\theta_i^U,X(\theta_i,\theta_{-i}))$.
We have
reached a contradiction. We conclude that $\chi$ has to satisfy PIA, which implies that $\chi$ is FNPW. Hence, $X$
is FNPW-implementable.  \end{proof}

Proof of Proposition~\ref{prop:LDS}:

\begin{proof}
The general LDS mechanism is rather complicated. Instead of describing LDS in
its general form, we focus on a specific LDS mechanism for three items, which
is characterized by reserve price $1$ and two levels: $[\{(A,B,C)\}]$ and
$[\{(A,B),(C)\}, \{(A),\\
  (B,C)\}]$. The mechanism works as follows. If there are
at least two agents whose valuations for $\{A,B,C\}$ are at least $3$, then we
combine $\{A,B,C\}$ into one bundle, and run the Vickrey auction.  If every
agent's valuation for $\{A,B,C\}$ is less than $3$, then we do the following.
We first introduce a dummy agent into the system.  The dummy agent has an
additive valuation function and values every item at $1$.  We only allow five
types of allocations: 1) The dummy agent wins everything.  2) The dummy agent
wins one of $\{A,B\}$ and $\{C\}$, and a non-dummy agent wins the other.  3)
The dummy agent wins one of $\{A\}$ and $\{B,C\}$, and a non-dummy agent wins
the other.  4) A non-dummy agent wins one of $\{A,B\}$ and $\{C\}$, and another
non-dummy agent wins the other.  5) A non-dummy agent wins one of $\{A\}$ and
$\{B,C\}$, and another non-dummy agent wins the other.  We run the VCG
mechanism on this restricted set of possible allocations.  Finally, if there is
only one agent whose valuation for $\{A,B,C\}$ is at least $3$, then this agent
is the only winner. She has the option to purchase all the items at price $3$,
or to obtain the result she would have obtained if everyone (including the
dummy agent) were to join in the above maximal-in-range mechanism. 

We only need to prove that the above specific LDS mechanism does not satisfy
withdrawal-monotonicity.  We consider the following scenario involving only
types that are single-minded.  $\theta_{-i}$ contains only one type from an
agent who bids $2.2$ on $\{A,B\}$.  If $\theta_i$ is bidding $1.3$ on $\{A\}$,
	  then $X(\theta_i,\theta_{-i})=\{A\}$.  If $\theta_i^U$ is bidding $1.05$
	  on $\{A\}$, and $\theta^a$ is bidding $2.9$ on $\{B,C\}$, then
	  $X(\theta_i^U,\theta_{-i}\cup\theta^a)=\{A\}$.
	  If $\theta_i^L$ is bidding $1.1$
	  on $\{A\}$, then
	  $X(\theta_i^L,\theta_{-i})=\emptyset$. That is, 
	  $v(\theta_i^L,X(\theta_i^L,\theta_{-i}))=0$.
	  According to withdrawal-monotonicity, we should have
	  $v(\theta_i^L,\{A\})=1.1\le v(\theta_i^U,\{A\})=1.05$, which is a contradiction.
  We conclude
  that, in general, LDS does not satisfy withdrawal-monotonicity, and hence is not FNPW.
  \end{proof}

Proof of Theorem~\ref{th:submodularity}:

\begin{proof} We first prove that if the type space satisfies submodularity,
then the VCG mechanism is FNPW.
We consider agent $i$.  Let $K$ be the set of false-name identities $i$ submits
and keeps at the end.  Let $W$ be the set of false-name
identities $i$ submits and withdraws.  We already
know that submodularity is sufficient for the VCG mechanism to be FNP. Hence,
if $K$ contains multiple identities, then $i$ might as well replace all of them
by one identity that reports $i$'s true type. We then show that the identities
in $W$ do not help $i$.
We use $S$ to denote the set of items won by $i$ at
the end.  To win $S$, $i$ pays the VCG price $U(G,\{-i\}\cup W)-U(G-S,\{-i\}\cup W)$
($\{-i\}$ is the set of other agents).  We use $S'$ to denote the set of items
won by identities in $W$, when we allocate items in $G-S$ to identities in
$\{-i\}\cup W$ efficiently.  We have that $U(G,\{-i\}\cup W)-U(G-S,\{-i\}\cup W)=
U(G,\{-i\}\cup W)-U(G-S-S',\{-i\})-U(S',W)\ge
U(G-S',\{-i\})+U(S',W)-U(G-S-S',\{-i\})-U(S',W)=
U(G-S',\{-i\})-U(G-S-S',\{-i\})$.  The submodularity condition implies that
$U(G-S',\{-i\})-U(G-S-S',\{-i\})\ge U(G,\{-i\})-U(G-S,\{-i\})$. But, the expression
on the right-hand side of the inequality is the price $i$ would be charged for
$S$ when she uses a single identifier.
That is, 
$i$ does not benefit from the false-name identities in $W$.
Therefore, the VCG mechanism is FNPW
if the type space satisfies submodularity.

Next, we prove that if the VCG mechanism is FNPW, then the type space must
satisfy submodularity (if it contains the additive valuations).  
Let $S$ be an arbitrary set of items. Let $i$ be an agent that
is interested in $S$. Since we allow additive valuations, such $i$ always exists ({\em e.g.}, $i$ may have a very large valuation for every item in $S$).
If $i$ bids truthfully, then
she can win $S$ at a price of $U(G,\{-i\})-U(G-S,\{-i\})$.  Let $S'$ be another
arbitrary set of items that does not intersect with $S$.  For each item $j$ in
$S'$, we introduce a false-name identity that is only interested in item $j$,
with value $c$,
where $c$ is set to a very large value ({\em e.g.}, larger than $U(G,\{-i\})$).
These false-name identities are allowed since we assume the type space contains
the additive valuations.  Let $W$ be the set of identities introduced.  With
$W$, $i$ can win $S$ at a price of $U(G,\{-i\}\cup W)-U(G-S,\{-i\}\cup W)$.
We have that $U(G,\{-i\}\cup W)-U(G-S,\{-i\}\cup W)
=U(G-S',\{-i\})+U(S',W)-U(G-S-S',\{-i\})-U(S',W)
=U(G-S',\{-i\})-U(G-S-S',\{-i\})$.  The new price should never be smaller than
the old price. Otherwise, there is an incentive for $i$ to submit false-name
bids and withdraw them.
That is, we have $U(G,\{-i\})-U(G-S,\{-i\})\le
U(G-S',\{-i\})-U(G-S-S',\{-i\})$.  Let $S_1=G-S$, $S_2=G-S'$, and $Y=\{-i\}$.
We have $U(S_1\cap S_2,Y)-U(S_1,Y)\le U(S_2,Y)-U(S_1\cup S_2,Y)$. Since
$S_1,S_2$, and $Y$ are arbitrary, we have submodularity.
\end{proof}

Proof of Theorem~\ref{th:1overm}:

\begin{proof} 
Let $\chi$ be the price function that corresponds to an FNPW mechanism with
optimal worst-case ratio.  Since the Set mechanism is FNPW, $\chi$'s worst-case
efficiency ratio is at least $\frac{1}{m}$.  We denote item $i$ by $s_i$.  We
consider the following types: 

$\theta_a$: the type of an agent that is single-minded on the grand bundle,
with value $1$.

$\theta_i$ ($i=1,2,\ldots,m$): the type of an agent that is single-minded
on $s_i$, with value $1-\epsilon$. Here, $\epsilon$ is a small positive number.

{\em Scenario 1:} There are two agents. Agent $a$ has type $\theta_a$.
Agent $1$ has type $\theta_1$.

{\em Scenario 2:} There are two agents. Both agents have type $\theta_1$.

{\em Scenario 3:} There are $m+1$ agents. Agent $a$ has type $\theta_a$.
Agent $i$ has type $\theta_i$ for $i=1,2,\ldots,m$.

We first prove that in scenario 1, agent $a$
wins. We start with the special case of $m=1$.  If
$\chi(\{s_1\},\{\theta_1\})>1-\epsilon$, then we consider scenario 2.  In scenario
2, both agents can not afford the only item.  That is, the efficiency ratio is
$0$. Hence, we must have $\chi(\{s_1\},\{\theta_1\})\le 1-\epsilon$. That is, in
scenario 1, in the case of $m=1$, agent $a$ must win.  The IIG condition implies
that this is also true for cases with $m>1$.  

Since agent $a$ is the only winner in scenario 1, we have
$\chi(\{s_1\},
\{\theta_a\})\ge 1-\epsilon$ (otherwise, agent $1$ would win in
scenario $1$).  $\epsilon$ can be made arbitrarily close to $0$; hence,
$\chi(\{s_1\},\{\theta_a\})\ge 1$.

Finally, we consider scenario 3. The price agent $1$ faces for $s_1$ is
$\chi(\{s_1\},\{\theta_a\} \cup (\bigcup\limits_{j\ne 1}\{\theta_j\}))$. 
This price is at least $\chi(\{s_1\},\{\theta_a\})=1$. That is, agent $1$ does not
win in scenario 3. By symmetry over the items,
 agent $i$ does not win for all
$i=1,2,\ldots,m$.  The efficiency ratio in this scenario is then at
most
$\frac{1}{m(1-\epsilon)}$, which goes to
$\frac{1}{m}$ as $\epsilon$ goes to $0$.  \end{proof}

Proof of Proposition~\ref{prop:mmvip}:

\begin{proof} We first prove that MMVIP is feasible. We need to show that, with appropriate tie-breaking, MMVIP will never allocate the same item
	to multiple agents. 
	Let us suppose that under MMVIP there is a scenario in
	which two agents, $i$ and $j$, both win item $s$. Let $S_i$ and $S_j$ be the
	sets of other items (items other than $s$) that $i$ and $j$ win at the end,
	respectively.  Let $v_i=v(i,S_i\cup\{s\})-v(i,S_i)$. That
	is, $v_i$ is $i$'s marginal value for $s$.  Let
	$v_j=v(j, S_j\cup\{s\})-v(j, S_j)$. That is, $v_j$ is $j$'s
	marginal value for $s$.  If $v_i>v_j$, then $j$ has to pay at least $v_i$
	to win $s$, which is too high for her; $j$ is better off not
	winning $s$.  Similarly, if $v_i<v_j$, then $i$ is better off not winning $s$.  If
	$v_i=v_j$, then $i$ and $j$ both have to pay at least their marginal value
	for $s$ to win $s$.  That is, they are either indifferent between winning $s$ or
	not, or prefer not to win. 
	The only case that does not lead to a contradiction is where they are both indifferent;
any tie-breaking rule can resolve this
	conflict.

We then show that MMVIP is FNPW. By Proposition~\ref{prop:s-nsaw}, we only need to
prove that the price function $\chi$ that characterizes MMVIP satisfies S-NSAW.
Let $O$ be an arbitrary set of agents.  $\forall S_1,S_2\subseteq G$ with
$S_1\cap S_2=\emptyset$, we have $\chi(S_1,O)+\chi(S_2,O)=\chi(S_1\cup S_2,O)$,
because MMVIP uses item pricing.  Hence, DLB is satisfied.  $\forall S\subseteq
G$, for any agent $a$ that is not in $O$, $\chi(S,O\cup\{a\})=\sum\limits_{s\in
S}\chi(s,O\cup\{a\}) =\sum\limits_{s\in S}\max\limits_{j\in
O\cup\{a\}}\max\limits_{S'\subseteq G-\{s\}}\{v(j, S'\cup\{s\})-v(j, S')\} \ge
\sum\limits_{s\in S}\max\limits_{j\in O}\max\limits_{S'\subseteq
G-\{s\}}\{v(j, S'\cup\{s\})-v(j, S')\} =\sum\limits_{s\in S}\chi(s,O)
=\chi(S,O)$.  That is, PIA is also satisfied.  \end{proof}

Proof of Proposition~\ref{prop:mmvipvcg}:

\begin{proof} When the agents' valuations are additive, we have that MMVIP's
item price function satisfies $\chi(s,O)=$\\$\max\limits_{j\in
O}\max\limits_{S\subseteq G-\{s\}}\{v(j, S\cup\{s\})-v(j,S)\} =
\max\limits_{j\in O} v(j,\{s\})$.  Thus, MMVIP is equivalent to $m$ separate
Vickrey auctions (one Vickrey auction for each item), and hence to VCG (which
also corresponds to $m$ separate Vickrey auctions when the valuations are
additive).  \end{proof}

Proof of Proposition~\ref{prop:mmvipminimal}:

\begin{proof}
For the sake of contradiction, let us assume that the proposition is false. That
is, we assume that for every set of items $S$ and every set of agents $O$,
we have $\chi'(S,O) \le \chi(S,O)$. 
Since $\chi \neq \chi'$, we have that
there exists at least one set of items $S$ and one set of agents $O$ such
that $\chi'(S,O) < \chi(S,O)$.
Since $\chi'(S,O)=\sum_{s\in S}\chi'(s,O)$ and $\chi(S,O)=\sum_{s\in
  S}\chi(s,O)$, it follows that there exists $s\in S$ such that
$\chi'(s,O)<\chi(s,O)$.  By the definition of MMVIP, $\chi(s,O)$
corresponds to the maximal marginal value of some agent $j\in O$. That is,
there exists $S'\subset G$ with $s \notin S'$ such that $\chi(s,O)=v(j,
S'\cup \{s\})-v(j,S')$. We construct an agent $x$, whose valuation function
is additive.  Let $x$'s valuations of items not in $S'\cup \{s\}$ be
extremely high, so that $x$ wins all these items under both mechanisms
$\chi$ and $\chi'$.  (We recall that we assume consumer sovereignty for
FNPW mechanisms, so that $\chi, \chi' < \infty$ everywhere.)  Let $x$'s
valuation on $s$ be $\chi(s,O)-\epsilon$ (where $\epsilon$ is small enough
so that $\chi(s,O)-\epsilon > \chi'(s,O)$). Let $x$'s valuation of items in
$S'$ be $0$. When the set of agents consists of $x$ and the agents in $O$,
we have that $x$ wins all the items except for those in $S'$ under $M$.
Since $M$ is FNPW, we have $\chi'(s,O)\ge \chi'(s,\{j\})$.  That is,
when the set of agents consists of only $x$ and $j$, $x$ also wins all the
items except for those in $S'$ under $M$.  Also, under $M$, $j$ wins all of
$S'$, because for any $s' \in S'$, we have $\chi'(s', \{x\}) \leq \chi(s',
\{x\}) = 0$.  However, we then have that $\chi'(s,\{x\})\le
\chi(s,\{x\})=\chi(s,O)-\epsilon =v(j,S'\cup \{s\})-v(j,S')-\epsilon$, so
that $j$ would choose to also win $s$ when facing $x$ under $M$.  That is,
under $M$, when the set of agents consists of only $x$ and $j$, $s$ is won
by both agents, contradicting the assumption that $M$ is feasible.  Thus,
assuming that the proposition is false leads to a contradiction.
\end{proof}

\section{Automated FNPW Mechanism Design}
\label{sec:amd}

In this section, we propose an automated mechanism design 
(AMD) technique that
transforms any feasible mechanism into an FNPW mechanism.  In our setting, a
feasible mechanism is characterized by a price function $\chi$.  We start with
any $\chi$ that corresponds to a feasible mechanism ({\em e.g.}, the price
function of the VCG mechanism).  Our technique modifies $\chi$ so that it
satisfies S-NSAW, while maintaining feasibility.

We recall that for general combinatorial auction settings, there are three
known FNPW mechanisms (Set, MB, and MMVIP), and four known FNP mechanisms (the
aforementioned three mechanisms, plus LDS).  Though computationally expensive
(like many other AMD techniques in other contexts), this technique has the
potential to enlarge the set of known FNPW (FNP) mechanisms. By designing tiny
instances of FNPW mechanisms via automated mechanism design, we may get a
better understanding of the structure of FNPW mechanisms, from which we can
then conjecture FNPW mechanisms in analytical form. Later in this section, we
show that in a specific setting, by starting with the VCG mechanism, the AMD
technique produces exactly the MMVIP mechanism.  That is, had we not known the
MMVIP mechanism, the AMD technique could have helped us find it (though it just
so happened that we discovered MMVIP before the AMD technique).  It remains an
open question of whether new, general FNPW mechanisms can be found in this way.

Let $H:\Theta^k\to [0,\infty)$ be a function that maps any set of agents
  $O$ (more precisely, their reported types) to a nonnegative number
  $H(O)$.  For any feasible mechanism $\chi$, we define $\chi^H$ as
  follows:

\begin{itemize}
	\item For any set of agents $O$, $\forall \emptyset\subsetneq S\subseteq G$, $\chi^H(S,O)=\chi(S,O)
		+H(O)$. 
	\item For any set of agents $O$, $\chi^H(\emptyset,O)=\chi(\emptyset,O)=0$.
\end{itemize} 

That is, moving from $\chi$ to $\chi^H$, if we fix the reported types of
the other agents $O$, then we are essentially increasing the price of every
nonempty set of items by the same amount, while keeping the price of
$\emptyset$ at $0$.

\begin{lemma}\label{lm:gm-sma}\cite{Yokoo06:False} $\forall$ feasible $\chi$, $\forall H$, $\chi^H$ is feasible.  \end{lemma}

This lemma was first proved in~\cite{Yokoo06:False}.\footnote{The GM-SMA
mechanism~\cite{Yokoo06:False} relies on this property. However, it has been
shown that GM-SMA is {\em not} FNP in \cite{Todo09:Characterizing}.} An agent
is allocated her favorite set of items (the set that maximizes valuation minus
payment) in (PORF) mechanism $\chi$.  From the perspective of agent $i$, the
set of types reported by the other agents $\theta_{-i}$ is fixed. That is, for
$i$, under $\chi^H$, the price of every nonempty set of items is increased by
the same amount $H(\theta_{-i})$.  Hence, agent $i$'s favorite set of items is
either unchanged, or has become $\emptyset$ (if $H(\theta_{-i})$ is too large).
It is thus easy to see that if $\chi$ never allocates the same item to more
than one agent, then neither does $\chi^H$.  That is, feasibility is not
affected.\footnote{ If the agents are single-minded, then in a PORF mechanism,
as long as the prices of larger sets of items are more expensive, an agent's
favorite set of items is either the set on which she is single-minded, or the
empty set.  Thus, we do not need to increase the price of every set by the same
amount. As long as we are increasing the prices, an agent's favorite set either
remains unchanged, or becomes empty (if the price increase on the set on which
she is single-minded is too high).  That is, for single-minded agents, we have
more flexibility in the process of transforming a feasible mechanism into an
FNPW mechanism.  Due to space constraint, we do not pursue this further here.}

\begin{theorem} $\forall$ feasible $\chi$, we define the following $H$.  For
	any set of agents $O$, $H(O)$ equals the maximum of the following two values:
	\begin{itemize}
		\item
	$\max\limits_{S_1, S_2\subseteq G, S_1\cap S_2=\emptyset}
		\{\chi(S_1\cup S_2,O)-\chi(S_1,O)-\chi(S_2,O)\}$
	\item
	$\max\limits_{\emptyset\subsetneq S\subseteq G, j\in O}\{\chi(S,O-\{j\})+H(O-\{j\})-\chi(S,O)\}$
\end{itemize}

We have that $\chi^H$ is FNPW.
\label{th:amd}
\end{theorem}

It should be noted that, for any $O$, the first expression in the theorem
is at least $0$ (setting $S_1=S_2=\emptyset$).  That is, $H$ never takes
negative values.  $\chi^H$ is feasible by Lemma~\ref{lm:gm-sma}.


\begin{proof}
	We prove that $\chi^H$ satisfies S-NSAW. By Proposition~\ref{prop:s-nsaw},
        this suffices to show that $\chi^H$ is FNPW.

	{\em Proof of DLB:} Let $O$ be an arbitrary set of agents.
        $\forall S_1,S_2\subseteq G$ with $S_1\cap S_2=\emptyset$, we prove
        that $\chi^H(S_1,O)+\chi^H(S_2,O)\ge \chi^H(S_1\cup S_2,O)$.  If at
        least one of $S_1$ and $S_2$ is empty, then w.l.o.g., we assume
        $S_1=\emptyset$.  In this case,
        $\chi^H(S_1,O)+\chi^H(S_2,O)=\chi^H(S_2,O)=\chi^H(S_1\cup S_2,O)$.
        If neither $S_1$ nor $S_2$ is empty, then we have
        $\chi^H(S_1,O)+\chi^H(S_2,O)-\chi^H(S_1\cup
        S_2,O)=H(O)+\chi(S_1,O)+\chi(S_2,O)-\chi(S_1\cup S_2,O) \ge
        H(O)-\max\limits_{S'_1\cap S'_2=\emptyset} \{\chi(S'_1\cup
        S'_2,O)-\chi(S'_1,O)-\chi(S'_2,O)\}\ge 0$.
		
	{\em Proof of PIA:} Let $O$ be an arbitrary set of agents. Let $a$
        be an agent that is not in $O$.  If $S$ is empty, then we have
        $\chi^H(S,O\cup\{a\})=\chi^H(S,O)=0$.
		$\forall \emptyset\subsetneq S\subseteq G$,
        $\chi^H(S,O\cup\{a\})=H(O\cup\{a\})+\chi(S,O\cup\{a\}) \ge
        (\chi(S,O)+H(S,O)-\chi(S,O\cup\{a\}))+\chi(S,O\cup\{a\})=\chi^H(S,O)$.
\end{proof}

This still leaves the question of how to compute the $H$ described in the
theorem; we address this next. Given $\chi$, for any agent $i$ and any set of
other types $\theta_{-i}$, we compute $H(\theta_{-i})$ using the following
dynamic programming
algorithm. 

\begin{center}
\framebox{\parbox{4in}{
For $t=0,1,\ldots,|\theta_{-i}|$

\hspace{0.1in} For any $T\subseteq \theta_{-i}$ with $|T|=t$

\hspace{0.2in} $h_1=\max\limits_{S_1,S_2\subseteq G, S_1\cap
S_2=\emptyset}\{\chi(S_1\cup S_2,T)-\chi(S_1,T)-\chi(S_2,T)\}$.  

\hspace{0.2in} $h_2=\max\limits_{\emptyset\subsetneq S\subseteq G, j\in
T}\{H(T-\{j\})+\chi(S,T-\{j\})-\chi(S,T)\}$.  

\hspace{0.2in} $H(T)=\max\{h_1,h_2\}$.
}}
\end{center}

\begin{proposition}
If we apply the AMD technique to a mechanism that already satisfies
S-NSAW, the mechanism remains unchanged.
\label{prop:same}
\end{proposition}

We use the phrase ``the AMD mechanism'' to denote the mechanism generated
by the AMD technique starting from VCG (though the AMD technique is not
restricted to starting from VCG).  
Next, we prove a proposition that is similar
to Proposition~\ref{prop:mmvipvcg}.

\begin{proposition} When we restrict the preference domain to additive valuations,
  the MMVIP, the VCG, and the AMD mechanism all coincide.  
\label{prop:allcoincide}
\end{proposition}


\begin{proof} 
 Proposition~\ref{prop:mmvipvcg} already shows that MMVIP and VCG coincide.  All
 that remains to show is that VCG already satisfies S-NSAW, so that by
 Proposition~\ref{prop:same}, AMD is also the same.  When the agents' valuations
 are additive, the VCG mechanism's price function $\chi$ is defined as
 follows: for any set of items $S\subset G$ and any set of additive agents
 $O$, $\chi(S,O)=\sum_{s\in S}x^s$, where $x^s$ is the highest valuation
 for item $s$ among the agents in $O$. It is easy to see that $\chi$
 satisfies S-NSAW.
\end{proof}

Moreover, the next proposition shows that in settings with exactly two
substitutable items, the AMD mechanism coincides with MMVIP (but not with
VCG).

\begin{proposition} In settings with exactly two substitutable items, the AMD mechanism
coincides with MMVIP.  
\label{prop:subcoincide}
\end{proposition}


\begin{proof} 
The proof is by induction on the number of agents.  When there is only one
agent, this agent faces price $0$ for every bundle under the VCG mechanism.
This already satisfies S-NSAW, so by Proposition~\ref{prop:same}, we do not need to
increase any price in the AMD process. Therefore, when $n=1$, the AMD
mechanism allocates all the items to the only agent for free.  The MMVIP
mechanism does the same.  Hence, when $n=1$, the AMD mechanism coincides
with MMVIP. For the induction step, we assume that the two mechanisms
coincide when $n \leq k$. When $n=k+1$, the price function of the VCG
mechanism is defined as: $\chi(\{A\},O)=v_{AB}^*-v_B^*$,
$\chi(\{B\},O)=v_{AB}^*-v_A^*$, and $\chi(\{AB\},O)=v_{AB}^*$.  Here, $A$
and $B$ are the two items.  $v_A^*$ is the highest valuation for $A$ by the
agents in $O$.  $v_B^*$ is the highest valuation for $B$ by the agents in
$O$.  $v_{AB}^*$ is the highest combined valuation for $\{A,B\}$ by the
agents in $O$ (which may be obtained by splitting the items across two
different agents, or giving both to the same agent).  Since the items are
substitutable, $v_{AB}^*\le v_A^*+v_B^*$.  Equivalently,
$\chi(\{A\},O)+\chi(\{B\},O)\le \chi(\{AB\},O)$.  Therefore, in the AMD
technique, the price of every bundle has to increase by at least
$\chi(\{A,B\},0)-\chi(\{A\},O)-\chi(\{B\},O)$.  That is, under the AMD
mechanism, the price of $A$ is at least $v_A^*$, the price of $B$ is at
least $v_B^*$, and the price of $\{A,B\}$ is at least $v_A^*+v_B^*$.  These
prices are high enough to guarantee the PIA condition, because by the
induction assumption, the AMD mechanism coincides with MMVIP for $n \leq
k$; so, it follows that the AMD technique results in exactly these prices.
They coincide with the prices under the MMVIP mechanism.  Therefore, by
induction, the AMD mechanism coincides with the MMVIP mechanism for any
number of agents, when there are exactly two substitutable
items.  \end{proof}

It remains an open question whether there are more general settings in
which the AMD mechanism and the MMVIP mechanism coincide.

Finally, we compare the revenue and allocative efficiency of the VCG mechanism,
the Set mechanism\footnote{The MB mechanism and the Set mechanism coincide in
our experimental setup (the whole bundle is every agent's minimal bundle).}, the MMVIP mechanism, and the AMD mechanism.  It should be noted that
the VCG mechanism is not FNPW in general. We use it as a benchmark.

We consider a combinatorial auction with two items $\{A,B\}$ and five
agents $\{1,2,\\
	   \ldots,5\}$.\footnote{We only focused on these tiny
  auctions because the AMD technique is computationally quite expensive.
  Nevertheless, even the solutions to tiny auctions can be helpful in
  conjecturing more general mechanisms.}
We denote agent $i$'s valuation
for set $S\subseteq \{A,B\}$ by $v_i^S$.  We consider two scenarios, one with
valuations displaying substitutability, and the other with valuations
displaying complementarity.  We randomly generate $1000$ instances for each
scenario.

{\em Valuations with substitutability:} The $v_i^{\{A\}}$ and the
$v_i^{\{B\}}$ are drawn independently from $U(0,1)$ (the uniform distribution
from $0$ to $1$).  For all $i$, $v_i^{\{A,B\}}$ is drawn independently from
$U(\max\{v_i^{\{A\}},v_i^{\{B\}}\},
v_i^{\{A\}}+v_i^{\{B\}})$.  In this scenario, AMD and MMVIP coincide.
  They perform better than the Set mechanism, both in terms of revenue and allocative efficiency.

\begin{center}
\begin{tabular}{|c|c|c|c|c|}
	\hline
	& VCG & Set & AMD & MMVIP \\
	\hline
	Revenue & $1.285$ & $1.002$ & $1.221$ & $1.221$ \\
	\hline
	Efficiency & $1.668$ & $1.236$ & $1.550$ & $1.550$ \\
	\hline
\end{tabular}
\end{center}

{\em Valuations with complementarity:} The $v_i^{\{A\}}$ and the $v_i^{\{B\}}$
are still drawn independently from $U(0,1)$.  For all $i$, $v_i^{\{A,B\}}$ is
set to be $(v_i^{\{A\}}+v_i^{\{B\}})(1+x_i)$, where the $x_i$ are also drawn
independently from $U(0,1)$.  It turns out that, in this scenario, Set performs
better than AMD and MMVIP, both in terms of revenue and allocative
efficiency.  (MMVIP performs especially poorly when valuations exhibit
complementarity, because every item can potentially have a very large marginal
value to another agent, leading to prices that are too high.) 

\begin{center}
	\begin{tabular}{|c|c|c|c|c|} \hline & VCG & Set & AMD & MMVIP \\ \hline
		Revenue & $1.864$ & $1.849$ & $1.288$ & $0.594$ \\ \hline Efficiency &
		$2.372$ & $2.365$ & $1.565$ & $0.721$ \\ \hline \end{tabular}
	\end{center} 
	
Thus, when there are two items and five agents, among these FNPW mechanisms, it seems that Set is
most desirable if it likely that there is significant complementarity, and AMD
is most desirable if it is likely that there is substitutability. (We cannot
use the VCG mechanism unless we are certain that the type space makes VCG
FNPW.)

\section{Characterizing FNP(W) in Social Choice Settings} 
\label{sec:social}

Throughout the paper, we have only discussed combinatorial auctions. In
this section, we focus on FNP(W)\footnote{In these settings, it does not
  matter whether withdrawal is allowed or not.} in social choice settings
(without payments).  Specifically, we present a characterization of FNP(W)
social choice functions (without payments). A social choice function $f$ is
defined as $f: \{\emptyset\} \cup \Theta \cup \Theta^2 \cup \ldots
\rightarrow \Omega$, where $\Theta$ is the space of all possible types of
an agent, and $\{\emptyset\} \cup \Theta \cup \Theta^2 \cup \ldots$ is the
space of all possible profiles (since we do not know how many agents there
are). $\Omega$ is the outcome space. Let agent $i$'s type be
$\theta_i$. Let the types of agents other than $i$ be $\theta_{-i}$. $i$'s
valuation for outcome $\omega \in \Omega$
is denoted by $v_i(\theta_i,\omega)$.\\

First, we present the following straightforward characterization of
strategy-proof social choice functions.

\begin{proposition}
A social choice function $f$ is strategy-proof if and only if it satisfies
the following condition: $\forall i,\theta_i,\theta_{-i}$, we have
$f(\theta_i,\theta_{-i})\in
\arg\max_{\theta_i'}v_i(\theta_i,f(\theta_i',\theta_{-i}))$.
\end{proposition}

\begin{proof}
If the above condition is satisfied, then $\forall
i,\theta_i,\theta_i',\theta_{-i}$, we have
$v_i(\theta_i,f(\theta_i,\theta_{-i}))\ge
v_i(\theta_i,f(\theta_i',\theta_{-i}))$. That is, reporting truthfully is a
dominant strategy.

If reporting truthfully is a dominant strategy, then $\forall
i,\theta_i,\theta_i',\theta_{-i}$, we have\\
$v_i(\theta_i,f(\theta_i,\theta_{-i}))\ge
v_i(\theta_i,f(\theta_i',\theta_{-i}))$. That is, $\forall
i,\theta_i,\theta_{-i}$, we have
$v_i(\theta_i,f(\theta_i,\theta_{-i}))\ge
\max_{\theta_i'}v_i(\theta_i,f(\theta_i',\theta_{-i}))$, which is
equivalent to $f(\theta_i,\theta_{-i})\in
\arg\max_{\theta_i'}v_i(\theta_i,f(\theta_i',\theta_{-i}))$.
\end{proof}

That is, an agent always receives her most-preferred choice among outcomes
that she can attain with some type report.  We are now ready to present the
characterization of FNP(W) social choice functions.

\begin{proposition}
Suppose that for every outcome $o \in \Omega$, there exists some type
$\theta_i \in \Theta$ such that $\{o\} = \\
\arg \max_{o' \in O}
u_{\theta_i}(o')$ (each $o$ is the unique most-preferred outcome for some
type).  Then, a strategy-proof and individually rational social choice
function $f$ is FNP(W) if and only if it satisfies the following condition:
$\forall i,\theta_{-i},\theta_0$, we have
$\{f(\theta_i,\theta_{-i})|\theta_i \in \Theta\}\supseteq
\{f(\theta_i,\theta_{-i}\cup \{\theta_0\})|\theta_i \in \Theta\}$.  That
is, with an additional other agent, the set of outcomes that an agent can
choose decreases or stays the same.
\end{proposition}

\begin{proof} We first show that if $f$ is FNP(W), then the condition must be
  satisfied. Suppose not, that is, for some $i, \theta_{-i},\theta_0$,
  there exists some $o \in \{f(\theta_i,\theta_{-i}\cup
  \{\theta_0\})|\theta_i \in \Theta\} \setminus
  \{f(\theta_i,\theta_{-i})|\theta_i \in \Theta\}$.  Then, by assumption,
  there exists some $\theta_i \in \Theta$ such that $\{o\} = 
  \arg \max_{o' \in O} u_{\theta_i}(o')$.  It follows that an agent facing type profile
  $\theta_{-i}$ cannot obtain $o$ with a single report, but can obtain it
  by reporting both $\theta_0$ and some other type (such as, by
  strategy-proofness, $\theta_{i}$).  Because $o$ is her unique
  most-preferred outcome, she prefers to engage in this manipulation,
  contradicting FNP(W).


Conversely, we show that if the  condition is satisfied, then $f$ is FNP(W).
By assumption, $f$ is strategy-proof and individually rational, so we only
need to check that an agent has no incentive to use multiple identifiers.
Suppose that $o$ is an outcome that $i$ can obtain when facing
$\theta_{-i}$ by submitting multiple identities.  Because the set of
choices is nonincreasing in the number of identifiers used according to the
condition, it must be that $o \in \{f(\theta_i,\theta_{-i})|\theta_i \in
\Theta\}$.  Hence, there is no reason for her to use more than one identity.
\end{proof}

\end{document}